\documentclass[12pt]{amsart} 
\usepackage[english]{babel}
\usepackage[utf8]{inputenc}
\usepackage{amsmath}
\usepackage{amssymb}
\usepackage{amsthm}
\usepackage{amsfonts}
\usepackage{mathtools}
\usepackage{graphicx}
\usepackage{subfigure}
\usepackage{bm}
\usepackage{tikz-cd}
\usepackage{mathrsfs}
\usepackage{xargs}
\usepackage[colorinlistoftodos]{todonotes}
\usepackage{enumitem}
\usepackage{yfonts}
\usepackage{ dsfont }
\usepackage{bbm}
\usepackage{geometry}
\usepackage{setspace}
\usepackage{pgfplots}
\usepackage{todonotes}

\usepackage{tikz}
\usetikzlibrary{patterns,arrows.meta,calc}
\usepackage{cite}
\pgfplotsset{compat=newest}
\usepgfplotslibrary{fillbetween}
\usetikzlibrary{positioning}
\usetikzlibrary{decorations.pathreplacing}
\usetikzlibrary{decorations,arrows}
\usetikzlibrary{decorations.markings}
\usetikzlibrary{patterns}
\usetikzlibrary{quotes,angles}
\usetikzlibrary{arrows}

\usepackage{hyperref}
\usepackage[noabbrev,capitalize]{cleveref}
\crefname{proposition}{Proposition}{Propositions}
\crefname{equation}{}{}

\newtheorem{theorem}{Theorem}[section]
\newtheorem{lemma}[theorem]{Lemma}
\newtheorem{proposition}[theorem]{Proposition}

\theoremstyle{definition}
\newtheorem{definition}[theorem]{Definition}

\newtheorem{remark}[theorem]{Remark}

\crefname{assumption}{Assumption}{Assumptions}
\crefname{definition}{Definition}{Definitions}
\crefname{corollary}{Corollary}{Corollaries}
\crefname{enumi}{item}{items}

\creflabelformat{subfigure}{#2\textsc{#1}#3}

\allowdisplaybreaks[4]

\newcommand*{\rom}[1]{\expandafter\@slowromancap\romannumeral #1@}
\newcommandx{\alex}[2][1=]{\todo[linecolor=red,backgroundcolor=red!25,bordercolor=red,#1]{Alex: #2}}

\newcommandx{\jiayu}[2][1=]{\todo[linecolor=red,backgroundcolor=red!25,bordercolor=red,#1]{Jiayu: #2}}

\DeclareMathOperator{\N}{\mathbb{N}}
\DeclareMathOperator{\Z}{\mathbb{Z}}

\DeclareMathOperator{\R}{\mathbb{R}}

\DeclareMathOperator{\tr}{tr}

\renewcommand{\tilde}{\widetilde}

\newcommand{\nm}{\noalign{\smallskip}}
\newcommand{\ds}{\displaystyle}

\renewcommand{\epsilon}{\varepsilon}

\let\emptyset\varnothing

\renewcommand{\tilde}{\widetilde}

\usepackage{tcolorbox}
\usetikzlibrary{tikzmark,calc}

\newcommand{\mc}[1]{\mathcal{#1}}

\newcommand{\abs}[1]{\left\lvert#1\right\rvert}

\newcommand{\norm}[1]{\left\lVert#1\right\rVert}

\newcommand{\ldz}{\{\rom{1},\rom{2}\}^{\mathbb{Z}}}
\newcommand{\ldm}{\{\rom{1},\rom{2}\}^M}

\newcommand{\rp}{\mathbb{RP}^1}

\DeclareMathOperator{\SL}{SL}

\title[Topological interface modes in aperiodic subwavelength resonator chains]{Topological interface modes in aperiodic subwavelength resonator chains}

\author{Habib Ammari} %
\address[H. Ammari]{Department of Mathematics, ETH Z\"{u}rich, R\"{a}mistrasse 101, CH-8092 Z\"{u}rich, Switzerland; Hong Kong Institute for Advanced Study, City University of Hong Kong, Kowloon Tong, Hong Kong}
\email{habib.ammari@math.ethz.ch}

\author{Jiayu Qiu} %
\address[J. Qiu]{Department of Mathematics, ETH Z\"{u}rich, R\"{a}mistrasse 101, CH-8092 Z\"{u}rich, Switzerland}
\email{jiayu.qiu@sam.math.ethz.ch}

 \author{Alexander Uhlmann}
\address[A. Uhlmann] {Department of Mathematics, ETH Z\"{u}rich, R\"{a}mistrasse 101, CH-8092 Z\"{u}rich, Switzerland}
 \email{alexander.uhlmann@sam.math.ethz.ch}

\date{}

\begin{document}

\begin{abstract}
We consider interface modes in block disordered subwavelength resonator chains in one dimension. Based on the capacitance operator formulation, which provides a first-order approximation of the spectral properties of dimer-type block resonator systems in the subwavelength regime, we show that a two-fold topological characterization of a block disordered resonator chain is available if it is of dominated type. The topological index used for the characterization is a generalization of the Zak phase associated with one-dimensional chiral-symmetric Hamiltonians. As a manifestation of the bulk-edge correspondence principle, we prove that a localized interface mode occurs whenever the system consists of two semi-infinite chains with different topological characters. We also illustrate our results from a dynamic perspective, which provides an explicit geometric picture of the interface modes, and finally present a variety of numerical results to complement the theoretical results.

\end{abstract}

\maketitle

\bigskip

\noindent \textbf{Keywords.} one-dimensional aperiodic chain, dimer-type block disordered system, topologically protected interface eigenmode, bandgap, uniform hyperbolicity; quasi-periodic system, Zak phase \par

\bigskip

\noindent \textbf{AMS subject classifications.} 
34L15 
35P20, 
37D20, 
37A30, 
47B36. 
35B34 
35J05  
82D03 
\\

\tableofcontents

\section{Introduction}

Subwavelength physics is concerned with wave interactions in materials with structures at subwavelength scales. Its ultimate goal is to manipulate waves at extremely low frequencies. Recent breakthroughs, such as the emergence of the field of metamaterials (\emph{i.e.}, microstructured materials with unusual localization and transport properties), have allowed us to do this in a way that is robust and overcomes traditional diffraction limits using high-contrast resonators; see, \emph{e.g.}, \cite{paa,sheng,fink,cbms}.

Recently, inspired by the rapid development of topological physics \cite{Klitzing80qhe,vonKlitzing2020forty_years,stone1992quantum,qizhang11topo_insulator,bernivig13topo_insulator}, great efforts have been made to realize topological phenomenon in the subwavelength regime. The ultimate goal is to achieve interface or edge modes between media with different topological characters, as a manifestation of the famous bulk-edge correspondence principle (BEC). The mathematical theory of BEC is well established; see \cite{avila2013shortrange+transfer,elgart2005shortrange+functional,graf2013shortrange+scattering,Kellendonk02landau+ktheory,ludewig2020shortrange+coarse,EG2002bec_discrete_functional,drouot2024bec_curvedinterfaces,de2016spectral_flow,Bra2019bec_discrete_K,BKR2017bec_discrete_K,AMZ2020bec_discrete_K,kubota2017bec_discrete_K,taarabt2014landau+functional,cornean2021landau+functional,CG2005bec_schrodinger_functional,DGR2011bec_schrodinger_functional,Kellendonk2004landau+ktheory,BR2018bec_continuous_K,gontier2023edge,bal2019dirac+functional,QB2024bec_dirac_functional,bal2022bec_dirac_functional_symbol,qiu2025bulk,lin2022transfer,drouot2021microlocal,thiang2023bec_inversion}. The most intriguing property of these interface modes is that, due to their topological origin, they exhibit strong robustness under imperfection and disorder present in the system, which leads to great potential for application in energy and information transportation. In addition to the aforementioned study on bulk-edge correspondence, the mathematical theory of topologically protected interface modes has been established in electronic systems \cite{fefferman2014topologically,fefferman2017topologically,fefferman2016honeycomb_edge,drouot2020edge} and classical wave systems \cite{ammari2024interface,ammari2020topological,lin2021mathematical,lee2019elliptic,Drouot2019TheBC,qiu2024mathematical,qiu2023mathematical,qiu2023bifurcation,li2024interface}; see also the references therein.

In this work, we study the interface modes in a one-dimensional block disordered subwavelength resonator chain. Specifically, we consider the eigenvalue and its associated eigenmode of the capacitance operator, which provides a first-order approximation of the spectral properties of one-dimensional acoustic resonance systems \cite{cbms}. Our first result (Proposition \ref{prop_twofold_quantization}) shows that a two-fold topological characterization of the block disordered resonator chain is available if it is of dominated type (see Definition \ref{def_dominant_system}). The topological index used for the characterization is a generalization of the Zak phase associated with one-dimensional chiral-symmetric Hamiltonians, which has been extensively studied (see, for instance, \cite{graf.shapiro2018BulkEdge,tauber2022chiral_finite_chain,thiang2023bec_inversion}). As a manifestation of the bulk-edge correspondence principle, we prove that the capacitance operator attains a localized interface mode whenever the system consists of two semi-infinite chains with distinct topological characters. We remark that, apart from this work, block disordered systems of subwavelength resonators have been studied in \cite{ammari.thalhammer.ea2025Uniform,ammari.barandun.ea2025Subwavelength,ammari.barandun.ea2025Universal}. Aside from the bulk-edge correspondence approach, we also illustrate our results based on the dynamical perspective (Section \ref{propmatrix}). Various numerical tests are provided to complement our theoretical results.

The remainder of this paper is organized as follows. In Section 2, we present the problem setting and our main theoretical results, including a topological characterization of block disordered subwavelength resonator chain (Proposition \ref{prop_twofold_quantization}) and the existence of topologically protected interface modes (Proposition \ref{thm_interface_mode_finite_chain}). The detailed proofs are presented in Section 4, the required preliminaries of which are listed in Section 3. In Section 5, we illustrate our results from a dynamical perspective (propagation matrix cocycle approach). Finally, in Section 6, we illustrate the potential applicability of our results to quasi-periodic structures: we consider a Fibonacci tiling and show numerically that, even though the Fibonacci tiling does not satisfy the pseudo-ergodicity condition, by concatenating two semi-infinite Fibonacci tilings with different topological characters interface modes occur.

\section{Set-up and main results}

We consider one-dimensional chains of subwavelength resonators consisting of dimer-type blocks, where each block is selected randomly from a library of two elements. The notation is fixed as follows. Denote by $Y_{k}:=[4k,4(k+1))$  the $k$\textsuperscript{th} block. We consider the scenario in which each block contains two resonators of the same length but with possibly different spacings. That is, we let $B_{k}^{\rom{1}/\rom{2}}\subset Y_k$ be the resonators inside the $k$\textsuperscript{th} cell, where
\begin{equation*}
B_{k}^{\rom{1}}:=(4k,4k+1)\cup(4k+2-\delta,4k+3-\delta),\quad B_{k}^{\rom{2}}:=(4k,4k+1)\cup(4k+2+\delta,4k+3+\delta),
\end{equation*}
with $\delta\in (0,1)$ controlling the different spacings of type-$\rom{1}$ and type-$\rom{2}$ blocks. The whole resonator chain is uniquely indexed by a sequence $\omega\in \{\rom{1},\rom{2}\}^{\mathbb{Z}}$: that is, each resonator system is defined as
\begin{equation*}
\ds  \mathcal{D}^{\omega}:=\bigcup_{k\in\mathbb{Z}}B_{k}^{\omega_k}.
\end{equation*}
In particular, we are interested in two cases in which type-$\rom{1}$ or type-$\rom{2}$ resonators are dominant.
\begin{definition} \label{def_dominant_system}
$\mathcal{D}^{\omega}$ is said to be \textit{right $\rom{1}-$dominant} (or \textit{left $\rom{1}-$dominant}, respectively) if
\begin{equation*}
    \lim_{K\to\infty}\frac{C_I(K;\omega)}{K+1}>\frac{1}{2},
\end{equation*}
(or $\lim_{K\to \infty}\frac{C_I(-K;\omega)}{K+1}>\frac{1}{2}$, respectively),
where for $d\in\{\rom{1},\rom{2}\}$
\[
    C_d(K;\omega) \coloneqq\begin{cases}
        \#\{j \mid 0\leq j\leq K, {\omega}_{j}=d\} & \text{if }K\geq0,\\
        \nm
        \#\{j \mid K\leq j< 0, {\omega}_{j}=d\} & \text{if }K<0 .\\
    \end{cases}
\]

In particular, we assume that the above limit exists.
$\mathcal{D}^{\omega}$ is said to be $\rom{1}-$dominant if it is $\rom{1}-$dominant both on the right and on the left. The $\rom{2}-$dominance is defined in a similar way. We say that $\mathcal{D}^{\omega}$ is dominated if it is either $\rom{1}-$ or $\rom{2}-$dominant.
\end{definition}
Apparently, the periodic systems, which correspond to $\omega_k\equiv \rom{1}$ and $\omega_k\equiv \rom{2}$, are dominated. 

In this paper, our main focus is on the spectrum of the capacitance operator associated with the resonator system. We note that the spectrum of capacitance operator provides an accurate approximation of the subwavelength resonant frequencies of coupled infinite systems of Helmholtz equations associated with high contrast-resonators; see, for instance,  \cite{ammari.barandun.ea2025Subwavelength, ammari.barandun.ea2025Universal, ammari.thalhammer.ea2025Uniform,ammari2024interface}. Let $S^{\rom{1}}_{l}=S^{\rom{2}}_{r}=1-\delta$ and $S^{\rom{1}}_{r}=S^{\rom{2}}_{l}=1+\delta$ be the distance between resonators in type-$\rom{1}$ and type-$\rom{2}$ blocks (the subscript $l/r$ refers to left/right, respectively). As shown in \cite{ammari2024interface}, the capacitance operator associated with the dimer-type resonator system $\mathcal{D}^{\omega}$ is defined as the following operator on $\mathcal{H}=\ell^2(\mathbb{Z})\otimes \mathbb{C}^2$:
\begin{equation*}
(\mathcal{C}^{\omega}\psi)_{k}:=
\begin{pmatrix}
0 & 0 \\ 
\nm \ds -\frac{1}{S^{\omega_{k}}_{r}} & 0
\end{pmatrix}\psi_{k+1}
+\begin{pmatrix}
0 & \ds -\frac{1}{S^{\omega_{k-1}}_{r}} \\ \nm 0 & 0
\end{pmatrix}\psi_{k-1}
+\begin{pmatrix}
\ds
\frac{1}{S^{\omega_{k}}_{l}}+\frac{1}{S^{\omega_{k-1}}_{r}} & \ds -\frac{1}{S^{\omega_{k}}_{l}} \\ 
\nm
\ds -\frac{1}{S^{\omega_{k}}_{l}} & \ds \frac{1}{S^{\omega_{k}}_{l}}+\frac{1}{S^{\omega_{k}}_{r}}
\end{pmatrix}\psi_{k} .
\end{equation*}
Clearly, $\mathcal{C}^{\omega}$ is bounded and self-adjoint for any configuration $\omega$. Notice, on the one hand, that $\mathcal{C}^{\omega}$ differs from the famous (disordered) Su–Schrieffer–Heeger model Hamiltonian in condensed matter physics only by a diagonal matrix (cf. \cite{thiang2023bec_inversion,graf2018bec_disorder_chiral,tauber2022chiral_finite_chain,perez2018ssh,inui1994unusual}). On the other hand, it is important to note that the second element of the diagonal matrix is constant for any $\omega$. In fact, we have
\begin{equation*}
\frac{1}{S^{\omega_{k}}_{l}}+\frac{1}{S^{\omega_{k}}_{r}}\equiv \frac{2}{1-\delta^2} .
\end{equation*}
This naturally leads to the following decomposition of $\mathcal{C}^{\omega}$:
\begin{equation} \label{eq_cap_operator_decom}
\mathcal{C}^{\omega}=\tilde{\mathcal{C}}^{\omega}+\frac{2}{1-\delta^2}\cdot \mathbbm{1}_{\mathcal{H}}+\mathcal{E}^{\omega}
\end{equation}
with
\begin{equation} \label{eq_off_diag_operator}
(\tilde{\mathcal{C}}^{\omega}\psi)_{k}:=
\begin{pmatrix}
0 & 0 \\ \nm \ds -\frac{1}{S^{\omega_{k}}_{r}} & 0
\end{pmatrix}\psi_{k+1}
+\begin{pmatrix}
0 & \ds -\frac{1}{S^{\omega_{k-1}}_{r}} \\ \nm  0 & 0
\end{pmatrix}\psi_{k-1}
+\begin{pmatrix}
0 & \ds -\frac{1}{S^{\omega_{k}}_{l}} \\  \nm 
\ds -\frac{1}{S^{\omega_{k}}_{l}} & 0
\end{pmatrix}\psi_{k} ,
\end{equation}
and 
\begin{equation} \label{eq_diag_perturbation}
(\mathcal{E}^{\omega}\psi)_{k}:=\begin{pmatrix}
\ds \frac{1}{S^{\omega_{k}}_{l}}+\frac{1}{S^{\omega_{k-1}}_{r}}-\frac{2}{1-\delta^2} & 0 \\ 
\nm 
0 & 0
\end{pmatrix}\psi_{k} .
\end{equation}
A critical reason for this decomposition is that the off-diagonal part is chiral symmetric, \emph{i.e.}, $\tilde{\mathcal{C}}^{\omega} \Gamma + \Gamma \tilde{\mathcal{C}}^{\omega} =0$ with $\Gamma=\mathbbm{1}\otimes \sigma_{z}$ being the sublattice operator (here $\sigma_{z}$ is the Pauli matrix). As being well-known, this leads to a topological classification of $\tilde{\mathcal{C}}^{\omega}$ (class A\rom{3} of the ten-fold table \cite{kitaev2009periodic_table}) if it is gapped at zero energy. This is verified when our resonator system is dominated.
\begin{proposition} \label{prop_gap_condition}
Suppose that $\mathcal{D}^{\omega}$ is either $\rom{1}-$ or $\rom{2}-$dominant. Then, the off-diagonal operator $\tilde{\mathcal{C}}^{\omega}$ has a spectral gap at $\lambda=0$, \emph{i.e.}, there exists $\Delta>0$ such that $$(-\Delta,\Delta)\cap \text{Spec} (\tilde{\mathcal{C}}^{\omega})=\emptyset.$$ Here, $\text{Spec}$ denotes the spectrum.
\end{proposition}

The proof is presented in Section \ref{sec_off_diag_gap}, which relies on the subordinacy theory of Jacobi operators, as recalled in Section \ref{sec_prelim_subordinacy}. 

One way of classifying $\tilde{\mathcal{C}}^{\omega}$ is based on the following tracial index, which generalizes the well-known Zak phase of one-dimensional gapped chiral-symmetric Hamiltonians \cite{graf.shapiro2018BulkEdge,tauber2022chiral_finite_chain,thiang2023bec_inversion,zak}. Since $\mathcal{C}^{\omega}$ is in one-to-one correspondence with the off-diagonal part $\tilde{\mathcal{C}}^{\omega}$, the following index is also associated with $\mathcal{I}(\mathcal{C}^{\omega})$:
\begin{equation} \label{eq_trace_bulk_index}
\mathcal{I}(\mathcal{C}^{\omega})=\mathcal{I}(\tilde{\mathcal{C}}^{\omega}):=\frac{1}{2}\text{tr}(\Gamma \tilde{\Pi}[\Lambda,\tilde{\Pi}]).
\end{equation}
Here, $[\, , \, ]$ denotes the commutator, $\tilde{\Pi}=\tilde{P}_{+}-\tilde{P}_{-}$ with $\tilde{P}_{\pm}:=\mathbbm{1}_{\mathbb{R}^{\pm}}(\tilde{\mathcal{C}}^{\omega})$ being the spectral projections of $\tilde{\mathcal{C}}^{\omega}$ to the particle/hole bands  and $\text{tr}$ denotes the trace. $\Lambda:\mathbb{Z}\to\mathbb{R}$ is a \textit{switch function}, \emph{i.e.}, $\Lambda(n)=1$ ($\Lambda(n)=0$, respectively) for sufficiently large and positive $n$ ($-n$, respectively), which is understood as a multiplication operator in \eqref{eq_trace_bulk_index}. By Proposition \ref{prop_gap_condition}, it is known that the bulk index \eqref{eq_trace_bulk_index} is quantized as integers (a detailed proof can be found in \cite[Section 3]{graf2018bec_disorder_chiral}; see also \cite{tauber2022chiral_finite_chain}).
\begin{proposition}[\cite{graf2018bec_disorder_chiral,tauber2022chiral_finite_chain}] \label{prop_indice_quantization}
Suppose that $\mathcal{D}^{\omega}$ is either $\rom{1}-$ or $\rom{2}-$dominant. Then, $\mathcal{I}(\mathcal{C}^{\omega})\in\mathbb{Z}$.
\end{proposition}
For our specific model, we further prove that the field of quantization is actually $\mathbb{Z}_2$.
\begin{proposition} \label{prop_twofold_quantization}
Under the assumption of Proposition \ref{prop_indice_quantization}, disordered chains of different dominance are topologically distinct in the sense that
\begin{equation*}
\mathcal{I}(\mathcal{C}^{\omega})
=\left\{
\begin{aligned}
&0 \quad \text{if $\mathcal{D}^{\omega}$ is $\rom{1}-$dominated,} \\
\nm
&1 \quad \text{if $\mathcal{D}^{\omega}$ is $\rom{2}-$dominated.}
\end{aligned}
\right.
\end{equation*}
\end{proposition}
The proof is given in Section \ref{sec_topo_indice}. We note that the main idea of proof follows a standard adiabatic deformation argument. Indeed, for any type-$\rom{1}$ or type-$\rom{2}$ dominated configuration $\omega$, since the zero-energy gap remains open if the dominance is preserved by Proposition \ref{prop_gap_condition} and the topological index is quantized, we can continuously deform $\omega$ to the periodic case, \emph{i.e.}, $\omega_{k}\equiv \rom{1}$ or   $\omega_{k}\equiv \rom{2}$, without changing the index. The proof is then complete by noting that the index \eqref{eq_trace_bulk_index} for periodic structures is equal to the well-known Zak phases \cite{zak}, which are explicitly calculated for our block model.

\begin{remark}
Proposition \ref{prop_twofold_quantization} indicates that a topological phase transition occurs when the configuration $\omega$ is equally distributed.
\end{remark}

The topological nature of the disordered chains is manifested by the appearance of interface modes when we attach two topologically distinct half-chains at the origin. The main result is stated as follows.
\begin{proposition} \label{thm_interface_mode_finite_chain}
Suppose that $\mathcal{D}^{\omega}$ is right $\rom{1}-$dominant and left $\rom{2}-$dominant. Then, we have $\lambda=\frac{2}{1-\delta^2}\in \text{Spec}_{pp}(\mathcal{C}^{\omega})$, with an associated eigenfunction that is localized near the origin. Here, $\text{Spec}_{pp}$ denotes the pure point spectrum.  
\end{proposition}

The proof is presented in Section \ref{sec_BIC}. We briefly sketch the idea as follows. The first step is to introduce an interface index associated with the off-diagonal part $\tilde{\mathcal{C}}^{\omega}$. This index characterizes the possible interface modes. Due to the chiral symmetry of $\tilde{\mathcal{C}}^{\omega}$, the interface index equals the difference of the bulk indices (of the two media aside the origin) as proved in \cite{graf2018bec_disorder_chiral}. (This equality is the mathematical formulation of the bulk-edge correspondence principle; we refer the interested reader to \cite{avila2013shortrange+transfer,elgart2005shortrange+functional,graf2013shortrange+scattering,Kellendonk02landau+ktheory,ludewig2020shortrange+coarse,EG2002bec_discrete_functional,drouot2024bec_curvedinterfaces,de2016spectral_flow,Bra2019bec_discrete_K,BKR2017bec_discrete_K,AMZ2020bec_discrete_K,kubota2017bec_discrete_K} for more studies on BEC in discrete systems.) Next, by examining the distribution of the eigenmodes of $\tilde{\mathcal{C}}^{\omega}$, it is shown that the eigenmodes decouple with the off-diagonal perturbation $\mathcal{E}^{\omega}$. This concludes the existence of interface modes of the original capacitance operator $\mathcal{C}^{\omega}$.

To illustrate the result in Proposition \ref{thm_interface_mode_finite_chain}, 
in Figure \ref{fig:interface_mode} we compute the spectra of $\mathcal{D}^{\omega}_K := \bigcup_{|k| \leq K} B_{k}^{\omega_k}$ for $K$ large enough that correspond to the two domination regimes and a truncated large resonator system. Depending on the domination regime, the interface mode lies on the lower or upper critical line. The frequency of the interface mode remains unchanged as long as the domination regime of $\omega$ is constant. (Nonetheless, the relative block probability has an influence on the decay rate of the corresponding interface mode.) This is discussed more closely in \cref{rmk:Interface_Mode_Decay} and is illustrated in \cref{fig:Interface_Mode_Decay}. 

\begin{figure}
    \centering
    \includegraphics[width=0.95\linewidth]{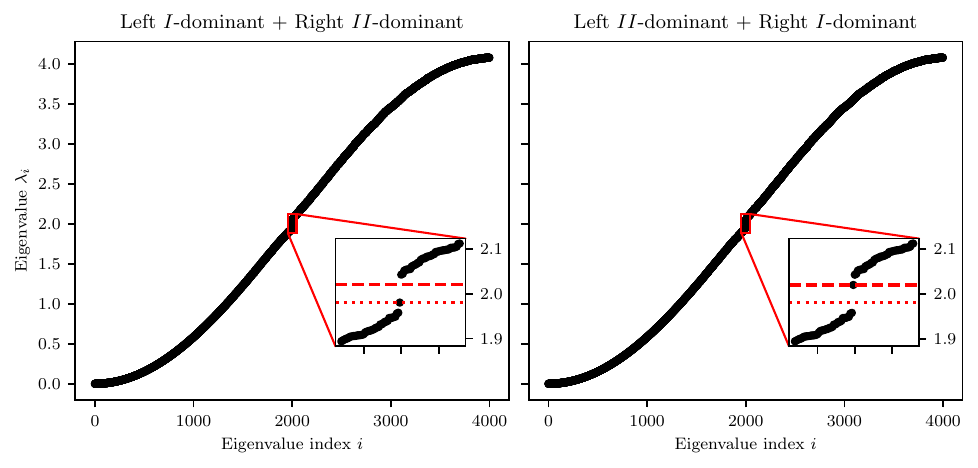}
    \caption{Spectrum of $\mathcal{D}^{\omega}$ for the two domination regimes ($K=2000, \delta=10^{-1}$ in both cases). The highlighted regions contain the interface modes together with the critical lines $\frac{8}{3+\sqrt{1+8\delta^2}}$ (red dots) and $\frac{2}{1-\delta^2}$ (red dashes). We can see that, depending on the domination regime, the interface mode lies on the lower or upper critical line.}
    \label{fig:interface_mode}
\end{figure}



\begin{remark}
The same argument works when $\mathcal{D}^{\omega}$ is left $\rom{1}-$dominant and right $\rom{2}-$dominant. In fact, by the bulk-edge correspondence (Proposition \ref{prop_bec}), one can still find an interface mode associated with the off-diagonal operator $\tilde{\mathcal{C}}^{\omega}$ with zero eigenvalue. However, the constant lifting of the eigenvalue is different from $\frac{2}{1-\delta^2}$ (one can show that it is equal to $\frac{8}{3+\sqrt{1+8\delta^2}}$; the details are left to the interested reader). Nonetheless, the interface modes are still decoupled from the diagonal disorder and preserved in $\text{Spec}(\mathcal{C}^{\omega})$; see Figure \ref{fig:interface_mode} for illustration.
\end{remark}


\begin{remark}
We note that it is difficult to determine the exact number of interface modes due to the disorder. Mathematically, one sees the diagonal perturbation $\mathcal{E}^{\omega}$ is of the same size compared to the spectral gap of the off-diagonal part $\tilde{\mathcal{C}}^{\omega}$, which not only breaks the chiral symmetry, but also leads to the failure of perturbation theory \cite{inui1994unusual,perez2018ssh}.
\end{remark}

\begin{remark}
We remark that the dominance in Definition \ref{def_dominant_system} is necessary to determine the interface modes. For example, the following configuration is apparently not right-dominated:
\begin{equation*}
\omega_k=
\left\{
\begin{aligned}
&1 \quad \text{for } 0\leq k\leq N-1 \text{ or } k=2n\quad (n\geq N), \\
\nm
&0 \quad \text{for } N\leq 2N-1 \text{ or } k=2n+1\quad (n\geq N) .
\end{aligned}
\right.
\end{equation*}
For $N$ being large enough, it is expected that $\lambda=\frac{2}{1-\delta^2}\in \text{Spec}_{pp}(\mathcal{C}^{\omega})$ whose eigenfunction peaks near $k=N$, not at the origin. This implies that $\mathcal{C}^{\omega}$ is not spectrally gapped and has a bulk-localized mode.
\end{remark}

\section{Preliminaries}

\subsection{Subordinacy theory of Jacobi operators} \label{sec_prelim_subordinacy}

We recall some basics of subordinacy theory on the classification of the spectrum of Jacobi operators. These results can be found in \cite{levi2023subordinacy}; see also \cite{koslover2005jacobi,gilbert1987subordinacy,gilbert1989subordinacy,jitomirskaya1996dimensional,christ1998absolutely} and the references therein. We begin with the definition of subordinate solutions.
\begin{definition}[Subordinate solution]
Let $\psi$ be a formal solution of $\tilde{\mathcal{C}}^{\omega}\psi=\lambda\psi$ supported on $\mathbb{Z}^{+}:=\{k\in\mathbb{Z}:k\geq 0\}$. It is said to be \textit{subordinate on $\mathbb{Z}^{+}$} if it is not identically zero and for any other linearly-independent solution $\varphi$ on $\mathbb{Z}^{+}$, it holds that
\begin{equation*}
\lim_{L\to\infty}\frac{\|\psi\cdot\mathbbm{1}_{[0,L]}\|_{\ell^2(\mathbb{Z})\otimes \mathbb{C}^2}}{\|\varphi\cdot\mathbbm{1}_{[0,L]}\|_{\ell^2(\mathbb{Z})\otimes \mathbb{C}^2}}=0.
\end{equation*}
The subordinance on $\mathbb{Z}^{-}:=\{k\in\mathbb{Z}:k\leq 0\}$ is defined similarly. A solution to $\tilde{\mathcal{C}}^{\omega}\psi=\lambda\psi$ supported on $\mathbb{Z}$ is said to be \textit{subordinate} if it is both subordinate on $\mathbb{Z}^{+}$ and on $\mathbb{Z}^{-}$.
\end{definition}

Decompose the spectral measure of $\tilde{\mathcal{C}}^{\omega}$ into singular and absolutely continuous parts (with respect to the Lebesgue measure), \emph{i.e.} $\tilde{\mu}_{s}^{\omega}$ and $\tilde{\mu}_{ac}^{\omega}$, respectively. The following theorem on the characterization of the spectral measure of $\tilde{\mathcal{C}}^{\omega}$ is our main tool to prove Proposition \ref{prop_gap_condition} (cf. \cite[Theorem 1.4]{levi2023subordinacy}). 
\begin{theorem} \label{thm_subordinacy}
The singular part of the spectral measure $\tilde{\mu}_{s}^{\omega}$ is supported on the set
\begin{equation} \label{eq_subordinacy_1}
E_{s}:=\{\lambda\in\mathbb{R}:\, \text{there exists a subordinate solution of $\tilde{\mathcal{C}}^{\omega}\psi=\lambda\psi$ on $\mathbb{Z}$}\}.
\end{equation}
On the other hand, $\tilde{\mu}_{ac}^{\omega}$ is supported on the set $E_{ac}=E_{+}\cup E_{-}$, where
\begin{equation} \label{eq_subordinacy_2}
E_{\pm}:=\{\lambda\in\mathbb{R}:\, \text{there exists no subordinate solution of $\tilde{\mathcal{C}}^{\omega}\psi=\lambda\psi$ on $\mathbb{Z}^{\pm}$}\}.
\end{equation}
\end{theorem}

\subsection{Some trace-class properties}

We list some basic trace-class properties of tight-binding operators. The proofs of these results are found in \cite{drouot2024bec_curvedinterfaces,qiu2025generalized,Marcelli2019spin_conductivity,elgart2005shortrange+functional}.

\begin{definition}[Tight-binding operator] \label{def_tight_bind}
A bounded operator $A$ on $\ell^2(\mathbb{Z})\otimes \mathbb{C}^2$ is called tight-binding (exponentially localized) if there exists $\alpha,C_{\alpha}>0$ such that the matrix elements of $A$ satisfy $\|A({x}, {y})\|\leq C_{\alpha}e^{-\alpha\|{x}-{y}\|}$ for any ${x}, {y}\in \ell^2(\mathbb{Z})$.
\end{definition}

The following property indicates that the commutator of a tight-binding operator with switch functions is trace-class.
\begin{proposition} \label{prop_commutator_trace_class}
For $A$ being a tight-binding operator and $\Lambda$ being a switch function, the commutator $[A,\Lambda]$ is trace-class.
\end{proposition}
As shown in the following proposition, another important property is the continuity of the trace under multiplication of a strongly convergent sequence. 
\begin{proposition} \label{prop_strong_continuity_trace}
Suppose that $A$ is trace-class and $B_n$ converges strongly to zero. Then, $\text{tr}(AB_n) \rightarrow 0$ as $n \rightarrow +\infty$.
\end{proposition}

\section{Topological classification of dominated disordered chains} 

\subsection{Zero-energy gap of off-diagonal operators: Proof of Proposition \ref{prop_gap_condition}}
\label{sec_off_diag_gap}

By Theorem \ref{thm_subordinacy}, it is sufficient to prove that $E_{s}=E_{\pm}$ does not contain $\lambda=0$. Then the existence of a spectral gap follows from the closedness of the spectrum. The key issue is to calculate the generalized eigenmode at $\lambda=0$.
\begin{proposition} \label{prop_generalized_mode}
Suppose that $\mathcal{D}^{\omega}$ is $\rom{1}-$dominant. Then, $\tilde{C}^{\omega}\psi=0$ admits two linearly independent formal solutions $\psi_1$ and $\psi_2$, which are supported on the first and second resonators of the dimer block, respectively,
\begin{equation} \label{eq_generalized_mode_1}
{e}_2\cdot(\psi_1)_{k}= {e}_1\cdot(\psi_2)_{k}=0.
\end{equation}
Moreover, $\psi_1$ and $\psi_2$ admit opposite asymptotics:
\begin{equation} \label{eq_generalized_mode_2}
\lim_{k\to\infty}\|(\psi_1)_{k}\|=\infty,\quad \lim_{k\to -\infty}\|(\psi_1)_{k}\|=0,
\end{equation}
\begin{equation} \label{eq_generalized_mode_3}
\lim_{k\to\infty}\|(\psi_2)_{k}\|=0,\quad \lim_{k\to -\infty}\|(\psi_2)_{k}\|=\infty .
\end{equation}
Reversely, when $\mathcal{D}^{\omega}$ is $\rom{2}-$dominant, the asymptotics of the two formal solutions $\psi_1$ and $\psi_2$ satisfying \eqref{eq_generalized_mode_1} are
\begin{equation} \label{eq_generalized_mode_4}
\lim_{k\to\infty}\|(\psi_1)_{k}\|=0,\quad \lim_{k\to -\infty}\|(\psi_1)_{k}\|=\infty ,
\end{equation}
and 
\begin{equation} \label{eq_generalized_mode_5}
\lim_{k\to\infty}\|(\psi_2)_{k}\|=\infty,\quad \lim_{k\to -\infty}\|(\psi_2)_{k}\|=0.
\end{equation}
\end{proposition}
\begin{proof}
We only prove the claim for the $\rom{1}-$dominated case. First of all, the existence of two linearly independent generalized eigenmodes of $\tilde{C}^{\omega}\psi=0$ follows from the fact that it is a second-order difference equation and is Hermitian. Now, suppose that $\psi$ is a generalized eigenmode. The equation $\tilde{C}^{\omega}\psi=0$ can be written elementwisely as
\begin{equation} \label{eq_generalized_mode_proof_1}
\psi_{k+1}^{(1)}=-\frac{S^{\omega_k}_{r}}{S^{\omega_k}_{l}}\psi_{k}^{(1)},\quad
\psi_{k}^{(2)}=-\frac{S^{\omega_{k}}_{l}}{S^{\omega_{k-1}}_{r}}\psi_{k-1}^{(2)} .
\end{equation}
This shows that the first and second components of $\psi$ are decoupled and implies \eqref{eq_generalized_mode_1}. To study the asymptotics at infinity, we perform the substitution
\begin{equation*}
\phi_k:=
\begin{pmatrix}
1 & 0 \\ \nm 0 & \ds \frac{1}{S^{\omega_k}_{r}}
\end{pmatrix}\psi_k .
\end{equation*}
This does not change the asymptotics at infinity, since $S^{\omega_k}_{r}\in \{1-\delta,1+\delta\}$ is uniformly bounded. Moreover, it is easy to check that $\phi_k$ satisfies the recurrence relation
\begin{equation} \label{eq_generalized_mode_proof_2}
\ds \phi_{k+1}^{(1)}=-\frac{S^{\omega_k}_{r}}{S^{\omega_k}_{l}}\phi_{k}^{(1)},\quad
\phi_{k+1}^{(2)}=-\frac{S^{\omega_{k+1}}_{l}}{S^{\omega_{k+1}}_{r}}\phi_{k}^{(2)} . 
\end{equation}
The advantage of employing this substitution is seen by comparing \eqref{eq_generalized_mode_proof_2} with \eqref{eq_generalized_mode_proof_1}. In fact, the dynamics of the second component is controlled solely by $w_k$ instead of the tuple $(w_k,w_{k-1})$; this greatly simplifies our analysis. 

Now, we prove \eqref{eq_generalized_mode_2}. Denote by $$\ds q_{\omega_{k}}:=-\frac{S^{\omega_k}_{r}}{S^{\omega_k}_{l}}.$$ Then, by definition, 
$$\ds q_{\rom{1}}=-\frac{1+\delta}{1-\delta}=(q_{\rom{2}})^{-1}.$$ By  assumption  (2), $\mathcal{D}^{\omega}$ is $\rom{1}-$dominant. We select $K^{\prime}>0$ so that, for all $K\geq K^{\prime}$,
\begin{equation*}
C_{\rom{1}}(K;\omega)-C_{\rom{2}}(K;\omega)>cK,
\end{equation*}
and
\begin{equation*}
C_{\rom{1}}(-K;\omega)-C_{\rom{2}}(-K;\omega)>cK,
\end{equation*}
for some $c>0$. This implies
\begin{equation*}
\ds |\phi_{K+1}^{(1)}|=|q_{\rom{1}}^{C_{\rom{1}(K;\omega)}}q_{\rom{2}}^{C_{\rom{2}(K;\omega)}}\phi_{0}^{(1)}|
\geq \big(\frac{1+\delta}{1-\delta}\big)^{cK}|\phi_{0}^{(1)}|,
\end{equation*}
and similarly
\begin{equation*}
\ds |\phi_{-K-1}^{(1)}|\leq \big(\frac{1+\delta}{1-\delta}\big)^{-cK}|\phi_{0}^{(1)}|.
\end{equation*}
Then \eqref{eq_generalized_mode_2} holds.

Following the same lines and using \eqref{eq_generalized_mode_proof_2}, we can check that the asymptotic behaviour of $\psi_2$ is reciprocal to that of $\psi_1$. This proves \eqref{eq_generalized_mode_3}.
\end{proof}

With Proposition \ref{prop_generalized_mode}, we are now ready to prove Proposition \ref{prop_gap_condition}; we will only consider the case that $\mathcal{D}^{\omega}$ is $\rom{1}-$dominated, and the proof of the $\rom{2}-$dominated case is similar. First, we observe that $\psi_2$ is a subordinate solution of $\tilde{\mathcal{C}}^{\omega}\psi=0$ on $\mathbb{Z}^{+}$ by Proposition \ref{prop_generalized_mode}. Hence, the set $E_{+}$ (defined in \eqref{eq_subordinacy_2}) is empty. Similarly, we see that $E_{-}=\emptyset$. By Theorem \ref{thm_subordinacy}, we know that the absolutely continuous spectrum of $\tilde{\mathcal{C}}$ is empty. On the other hand, the set $E_{s}$ is also empty. In fact, if $\psi$ is a subordinate solution of $\tilde{\mathcal{C}}^{\omega}\psi=0$ on $\mathbb{Z}$, it necessarily decays to zero as $k\to\infty$, and hence has a vanishing first component at each $k$ by Proposition \ref{prop_generalized_mode}; but then one sees that $\psi$ blows up as $k\to -\infty$ if it is not identically zero, leading to a contradiction. Hence, we conclude that $0\notin \text{Spec}(\tilde{\mathcal{C}}^{\omega})$.


\subsection{Bulk indices of gapped disordered chains: Proof of Proposition \ref{prop_twofold_quantization}} \label{sec_topo_indice}

We prove Proposition \ref{prop_twofold_quantization} for the case that $\mathcal{D}^{\omega}$ is $\rom{1}-$dominated in detail. The proof for the $\rom{2}-$dominated is sketched at the end of this subsection.

First, we claim that we can deform a $\rom{1}-$dominated structure by turning the type-$\rom{2}$ resonator into type-$\rom{1}$ without closing the spectral gap. The following result holds. 
\begin{lemma} \label{lem_deform}
Suppose that $\mathcal{D}^{\omega}$ is $\rom{1}-$dominated. We define the family of lengths
\begin{equation*}
S^{\omega_k}_{r,t}=
\left\{
\begin{aligned}
& S^{\omega_k}_{r},\quad \omega_{k}= \rom{1}, \\
& S^{\rom{2}}_{r}+2t\delta,\quad \omega_{k}= \rom{2}, \\
\end{aligned}
\right. \quad
S^{\omega_k}_{l,t}=
\left\{
\begin{aligned}
& S^{\omega_k}_{l},\quad \omega_{k}= \rom{1}, \\
& S^{\rom{2}}_{l}-2t\delta,\quad \omega_{k}= \rom{2}, \\
\end{aligned}
\right. \quad (0\leq t\leq 1). 
\end{equation*}
Define the operator $\tilde{\mathcal{C}}^{\omega}_{t}$ similarly as in \eqref{eq_off_diag_operator} with the new lengths $(S^{\omega_k}_{r,t},S^{\omega_k}_{l,t})_{k\in\mathbb{Z}}$. Then, there exists $\Delta>0$ such that $(-\Delta,\Delta)\cap \text{Spec}(\tilde{\mathcal{C}}^{\omega}_{t})=\emptyset$ for any $t\in [0,1]$.
\end{lemma}
\begin{proof}
The idea is the same as Proposition \ref{prop_gap_condition}. For any $t\in [0,1]$, we can directly check that the generalized eigenmodes of $\tilde{\mathcal{C}}^{\omega}_{t}\psi=0$ still satisfy the properties listed in Proposition \ref{prop_generalized_mode}. This is because the first and second components of the generalized modes stay decoupled when we range $t$ (check it by \eqref{eq_generalized_mode_proof_2}). Then, we conclude that $0\notin \text{Spec}(\tilde{\mathcal{C}}^{\omega}_{t})$ for $t\in [0,1]$ recalling the subordinacy theory. Finally, we note that the gap at any $t_0\in [0,1]$ persists for $|t-t_0|$ being small by the perturbation theory (recall that $\tilde{\mathcal{C}}^{\omega}_{t}$ depends continuously on $t$). Hence, we see that a spectral gap exists globally by a standard compactness argument applied on $t\in [0,1]$.
\end{proof}

Now, we prove Proposition \ref{prop_twofold_quantization}. Analogously to \eqref{eq_trace_bulk_index}, we define
\begin{equation*}
\mathcal{I}(\tilde{\mathcal{C}}^{\omega}_{t}):=\frac{1}{2}\text{tr}(\Gamma \tilde{\Pi}_{t}[\Lambda,\tilde{\Pi}_{t}])
\end{equation*}
with $\tilde{\Pi}_{t}=\tilde{P}_{+,t}-\tilde{P}_{-,t}$ with $\tilde{P}_{\pm,t}:=\mathbbm{1}_{\mathbb{R}^{\pm}}(\tilde{\mathcal{C}}^{\omega}_{t})$. It is now sufficient to prove the following statement:
\begin{equation} \label{eq_twofold_quantization_proof_1}
\mathcal{I}(\tilde{\mathcal{C}}^{\omega}_{t})-\mathcal{I}(\tilde{\mathcal{C}}^{\omega}_{t^{\prime}})=0 \quad \text{for any $t,t^{\prime}\in[0,1]$ being sufficiently close.}
\end{equation}
Denoting $\tilde{\Pi}_{t,t^{\prime}}=\tilde{\Pi}_{t}-\tilde{\Pi}_{t^{\prime}}$, we rewrite the above difference as
\begin{equation} \label{eq_twofold_quantization_proof_2}
\mathcal{I}(\tilde{\mathcal{C}}^{\omega}_{t})-\mathcal{I}(\tilde{\mathcal{C}}^{\omega}_{t^{\prime}})
=\frac{1}{2}\text{tr}(\Gamma \tilde{\Pi}_{t,t^{\prime}}[\Lambda,\tilde{\Pi}_{t}])+\frac{1}{2}\text{tr}(\Gamma \tilde{\Pi}_{t}[\Lambda,\tilde{\Pi}_{t,t^{\prime}}]).
\end{equation}
Recall that 1) $[\Lambda,\tilde{\Pi}_{t}]$ is trace-class since the spectral projection $\tilde{\Pi}_{t}$ decays exponentially (by Proposition \ref{prop_commutator_trace_class}) due to the gap condition in Lemma \ref{lem_deform}, and 2) $\|\tilde{\Pi}_{t,t^{\prime}}\|$ is small by the fact that $\|\tilde{\mathcal{C}}^{\omega}_{t}-\tilde{\mathcal{C}}^{\omega}_{t^{\prime}}\|=\mathcal{O}(|t-t^{\prime}|)$, the gap estimate in Lemma \ref{lem_deform}, and the standard perturbation theory. Hence, the first trace on the right of \eqref{eq_twofold_quantization_proof_2} is of order $\mathcal{O}(|t-t^{\prime}|)$ by Proposition \ref{prop_strong_continuity_trace}. Similar arguments work for the other trace. In conclusion, we find that
\begin{equation*}
\mathcal{I}(\tilde{\mathcal{C}}^{\omega}_{t})-\mathcal{I}(\tilde{\mathcal{C}}^{\omega}_{t^{\prime}})=\mathcal{O}(|t-t^{\prime}|).
\end{equation*}
This estimate, together with the quantization $\mathcal{I}(\tilde{\mathcal{C}}^{\omega}_{t})\in\mathbb{Z}$ (by the gap condition established in Lemma \ref{lem_deform} and the result of \cite[Section 3]{graf2018bec_disorder_chiral}), concludes the proof of \eqref{eq_twofold_quantization_proof_1}.

\begin{remark}
The gap condition is fundamental for the proof. We refer the reader to \cite{richter2002homotopy} for a more sophisticated deformation argument for quantized bulk indices, where the gap condition is replaced by a more general one, \emph{i.e.}, the mobility gap condition.
\end{remark}

\subsection{Bulk-interface correspondence} \label{sec_BIC}

Throughout this subsection, we fix $\mathcal{D}^{\omega}$ to be right $\rom{1}-$dominant and left $\rom{2}-$dominant, as supposed in Proposition \ref{thm_interface_mode_finite_chain}. The right/left bulk media are defined as
\begin{equation*}
\mathcal{D}^{\omega^R}\sim \omega^R=(\omega^R)_k:=\left\{
\begin{aligned}
&\omega_k,\quad k\geq 0, \\
&\omega_{-k},\quad k\leq 0,
\end{aligned}
\right. \quad
\mathcal{D}^{\omega^L}\sim \omega^L=(\omega^L)_k:=\left\{
\begin{aligned}
&\omega_k,\quad k\leq 0, \\
&\omega_{-k},\quad k\geq 0,
\end{aligned}
\right.
\end{equation*}
\emph{i.e.}, we simply reflect the half structure about the origin. The $\omega-$dependence will be omitted since the structure is fixed. We abbreviate the capacitance operator associated with the interface/right/left structure as $\mathcal{C}$, $\mathcal{C}^{R}$, and $\mathcal{C}^{L}$, respectively. The off-diagonal part of these operators, as introduced in \eqref{eq_off_diag_operator}, are denoted as $\tilde{\mathcal{C}}$, $\tilde{\mathcal{C}}^{R}$, and $\tilde{\mathcal{C}}^{L}$. Similarly, we denote the spectral projections as $\tilde{P}_{\pm}^{*}:=\mathbbm{1}_{\mathbb{R}^{\pm}}(\tilde{\mathcal{C}}^{*})$, with $*\in\{\text{R,L,none}\}$.

Let $\tilde{P}_{0}:=\mathbbm{1}_{\{0\}}(\tilde{\mathcal{C}})$ be the zero-energy spectral projection of the off-diagonal capacitance operator. We introduce the edge index as
\begin{equation} \label{eq_edge_index}
\mathcal{I}_{edge}(\tilde{\mathcal{C}}):=\text{tr}(\Gamma \tilde{P}_{0}).
\end{equation}
This is well-defined since the integral kernel of $\tilde{P}_{0}$ is exponentially localized. Indeed, we note that $\lambda=0$ is not included in the essential spectrum $\text{Spec}_{ess}(\tilde{\mathcal{C}})=\text{Spec}_{ess}(\tilde{\mathcal{C}^{R}})\cup \text{Spec}_{ess}(\tilde{\mathcal{C}}^{L})$, which is easily seen by the Weyl criterion and recalling that $\lambda\notin \text{Spec}(\tilde{\mathcal{C}}^{R})\cup\text{Spec}(\tilde{\mathcal{C}}^{L})$ as proved in Proposition \ref{prop_gap_condition}. Then the localization of $\tilde{P}_{0}$ follows directly from the Combes-Thomas estimate of eigenfunctions in the pure point spectrum of self-adjoint operators.

\begin{figure}
    \centering
    \includegraphics[width=0.95\linewidth]{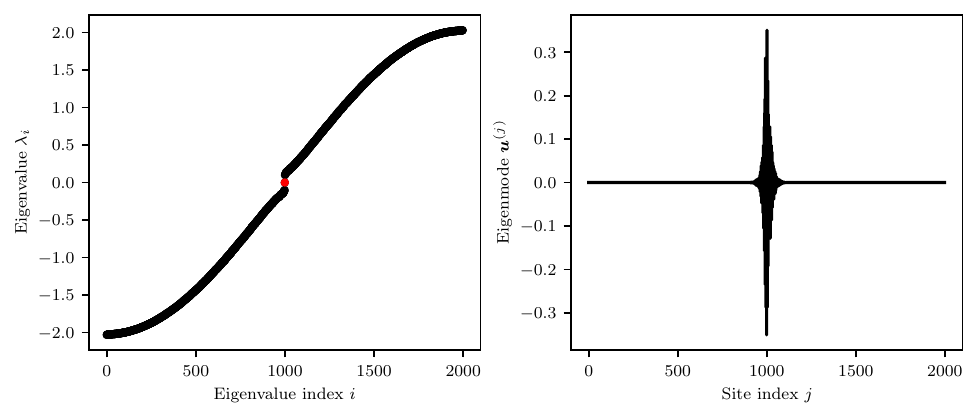}
    \caption{Spectrum and interface mode of the off-diagonal operator $\tilde{\mathcal{C}}^{\omega}$ for a left-$II$, right-$I$ dominated sequence $\omega$ ($K=1000, \delta=10^{-1}$). {Left:} Spectrum $\text{Spec}(\tilde{\mathcal{C}}^{\omega})$ with the interface eigenvalue at $\lambda=0$ highlighted in red. {Right:} The interface mode $u$ corresponding to $\lambda=0$. We can see that it is indeed exponentially localised at the interface $j=1000$.}
    \label{fig:Chiral_Spectrum}
\end{figure}

The bulk-edge correspondence indicates that the following result holds. 
\begin{proposition} \label{prop_bec}
Suppose that 1) $0\notin \text{Spec}(\tilde{\mathcal{C}}^{R})\cup \text{Spec}(\tilde{\mathcal{C}}^{L})$ and 2) $\tilde{P}_{-}$, $\tilde{P}_{-}^{R}$, and $\tilde{P}_{-}^{L}$ are tight-binding in the sense of Definition \ref{def_tight_bind}. Then, we have the following:
\begin{equation*}
\mathcal{I}_{edge}(\tilde{\mathcal{C}})=\mathcal{I}(\tilde{\mathcal{C}}^{R})-\mathcal{I}(\tilde{\mathcal{C}}^{L}).
\end{equation*}
\end{proposition}
The proof follows the same lines of \cite[Theorem 1]{graf2018bec_disorder_chiral}. Note that \cite[Theorem 1]{graf2018bec_disorder_chiral} is originally posed on a half line structure, which truncates a bulk medium at the origin. Nonetheless, the proof of bulk-edge correspondence readily carries to the interface structure, as we considered, without any essential change. We note that assumption 1) holds by Proposition \ref{prop_gap_condition}, which also gives the tight-bindingness of $\tilde{P}_{-}^{R/L}$. To see that $\tilde{P}_{-}$ is also tight-binding, we recall that $\lambda=0$ is at most an isolated eigenvalue of $\tilde{\mathcal{C}}$ as discussed at the beginning of this section. Hence, we see that $\tilde{P}_{-}=\mathbbm{1}_{(-\infty,-\epsilon)}(\tilde{\mathcal{C}})$ for some $-\epsilon\notin \text{Spec}(\tilde{\mathcal{C}})$, and can conclude that $\tilde{P}_{-}$ is tight-binding using again the Combes-Thomas estimate.

With Propositions \ref{prop_bec} and \ref{prop_twofold_quantization}, we conclude that $\mathcal{I}_{edge}(\tilde{\mathcal{C}})\neq 0$ and, therefore, $\tilde{P}_{-}\neq 0$. Let $\psi\in \text{Ran}(\tilde{P}_{-})$ be an eigenfunction associated with $\lambda=0$. By Proposition \ref{prop_generalized_mode}, it necessarily holds that ${e}_1\cdot\psi_{k}=0$ to ensure $\psi\in\ell^2(\mathbb{Z})\otimes \mathbb{C}^2$. This implies that $\psi$ decouples with the diagonal perturbation operator introduced in \eqref{eq_diag_perturbation}, \emph{i.e.}, $\mathcal{E}^{\omega}\psi=0$. Hence, we conclude that $\lambda=\frac{2}{1-\delta^2}$ is an eigenvalue of the original capacitance operator by recalling the decomposition \eqref{eq_cap_operator_decom}. This concludes the proof of Proposition \ref{thm_interface_mode_finite_chain}.

The above argument is illustrated in \cref{fig:Chiral_Spectrum} for a left-$II$, right-$I$ dominated sequence $\omega$ and a large truncated resonator system $\mathcal{D}^{\omega}_K$. The off-diagonal operator $\tilde{\mathcal{C}}^{\omega}$ indeed exhibits an interface eigenvalue at $0$ with an exponentially localised eigenmode.

\section{Propagation matrix cocycle perspective} \label{propmatrix}
Finally, our aim is to connect the results thus far with the block propagation matrix framework introduced in \cite{ammari.thalhammer.ea2025Uniform}. Doing so will provide an explicit geometric perspective for the observed interface modes.

We begin by introducing some crucial concepts and results from \cite{ammari.thalhammer.ea2025Uniform}. 

\begin{definition}[Propagation matrix]\label{def:propmat}
    Given $\delta>0$ and $\lambda\geq 0$ for the type-$I$ and type-$II$ dimers, we can define the \emph{block propagation matrices}
    \begin{equation}\label{eq:propmat}
        \mc P_I^\lambda = \begin{pmatrix}
            1-\lambda(3+\delta-(1-\delta^2)\lambda)&2-(1-\delta^2)\lambda\\
            -\lambda(2-(1-\delta)\lambda) & 1-(1-\delta)\lambda
        \end{pmatrix}\in \SL(2,\R),
    \end{equation}
   where $\mc P_{II}^\lambda$ is defined analogously under the replacement $\delta\to-\delta$ and $\SL(2,\R)$ is the set of real matrices with determinant $1$.

    Given any sequence $\omega\in \ldz$, this defines a \emph{block propagation matrix cocycle} via
    \begin{equation}
        \begin{aligned}
            \mc P^\lambda:\Z &\to \SL(2,\R)\\
            j &\mapsto \mc P^\lambda_{{\omega}_{j}} .
        \end{aligned}
    \end{equation}

    We further define the \emph{cocycle iteration} as 
    \begin{equation}\label{eq:cocycle_iteration}
        \mc P^\lambda_n(i) \coloneqq \begin{cases}
            \mc P^\lambda(i+n-1)\cdot \cdot \cdot \mc P^\lambda(i), & n\geq 1,\\
            I_2, & n=0,\\
            (\mc P^\lambda(i+n))^{-1}\cdot \cdot \cdot (\mc P^\lambda(i-1))^{-1}, & n\leq-1.
        \end{cases}
    \end{equation}
\end{definition}
Intuitively, the block propagation matrices $\mc P_{I}^\lambda$ and $\mc P_{II}^\lambda$ propagate a solution tuple $(u(x), u'(x))^\top$ across their respective block and we refer to \cite[Definition 3.8]{ammari.thalhammer.ea2025Uniform} for a more detailed introduction.

The reason we are interested in block propagation matrices is that their properties completely characterize the spectrum of $\mc C^\omega$, but before we can state this result, we must introduce two conditions: the \emph{pseudo-ergodicity} of the block sequence $\omega$ and the \emph{source-sink condition} on the block propagation matrices.

\begin{definition}\label{def:pseudoergodic}
    Let $\mathbb{I} \in \{\mathbb{Z},\mathbb{Z}_\pm\}$. We say that a sequence $\omega \in \{I, II\}^\mathbb{I}$ is \emph{pseudo-ergodic} if any finite sequence $\omega' \in \ldm$, where $M \in \mathbb{N}$, is contained in $\omega$.
    
    More precisely, for any $\omega' \in \ldm$, there must exist $j_0\in \Z$ such that ${\omega}_{j_0+j} = (\omega')_{j}$ for all $j=0,\ldots, M-1$. We will denote this by $\omega'\prec \omega$.
\end{definition}

Because $\mc P_d^\lambda\in \operatorname{SL}(2,\R)$, we can identify block propagation matrices with M\"obius transformations on the real projective space $\rp$, which in turn is diffeomorphic to the sphere $\rp\simeq \mathbb{S}^1$. In this language, $\mc P_d^\lambda$ is hyperbolic if and only if the corresponding M\"oebius transformation is hyperbolic. This transformation then has two distinct fixed points, one source and one sink, which we denote by $s(\mc P_d^\lambda)\in \rp$ and $u(\mc P_d^\lambda) \in \rp$, respectively\footnote{Namely, let $(\xi_1, E_1), (\xi_2, E_2)$ be the eigenvalues and eigenspaces of $\mc P_d^\lambda$ such that $\abs{\xi_1}<1<\abs{\xi_2}$. After identifying these one-dimensional eigenspaces with points in $\rp$, we find that $s(\mc P_d^\lambda)=E_1$ and $u(\mc P_d^\lambda)=E_2$.}.

The source-sink condition ensures the existence of an invariant cone for a family of hyperbolic block propagation matrices.
\begin{definition}[Source-sink condition]
    We say that the family of hyperbolic block propagation matrices $\{\mc P_1^\lambda, \ldots, \mc P_D^\lambda\}$ satisfies the \emph{source-sink condition (SSC)} if all the sink fixed points $u(P_d^\lambda),$ for $ d=1,\ldots, D,$ lie in the same connected component of $\rp\setminus\{s(P_1^\lambda), \ldots, s(P_D^\lambda)\}$.
\end{definition}

We are now in a position to characterize the spectrum of $\mc C^\omega$.
\begin{theorem}[\!{\cite[Theorem 3.20]{ammari.thalhammer.ea2025Uniform}}\;]\label{thm:saxonhutner}
    Let $\delta>0$ and consider a pseudo-ergodic block sequence $\omega \in \ldz$ and corresponding Jacobi operator $\mc C^\omega$. 
    
    Let $\lambda\geq 0$ and assume that both block propagation matrices $\mc P_d^\lambda, d=I, II$ are hyperbolic, \emph{i.e.},
    \[
        \abs{\tr \mc P_d^\lambda} > 2 \quad \text{ for both }d=I,II.
    \]
    Assume further that the family $\{\mc P_I^\lambda, \mc P_{II}^\lambda\}$ satisfies the \emph{source-sink condition}. Then, we must have 
    \[
        \lambda\notin  \text{Spec}(\mc C^\omega).
    \]

    Conversely, if we have $\abs{\tr \mc P_d^\lambda} \leq 2$ for any $d=I, II$, then $\lambda \in \text{Spec}(\mc C^\omega)$.
\end{theorem}

To find the spectrum of $\mc C^\omega$ it thus suffices to study the bandgaps (\emph{i.e.}, the intervals where $\abs{\tr \mc P_d^\lambda}>2$) and source-sink behaviour of the block propagation matrices. From \cref{eq:propmat} we can explicitly find the bandgaps  for both types of blocks. Namely, we find that
\begin{equation}
    \abs{\tr \mc P_d^\lambda}>2 \iff \lambda \in (\frac{2}{1+\delta}, \frac{2}{1-\delta})\cup (\frac{4}{1-\delta^2}, \infty)
\end{equation}
for both $d=I,II$. The fact that both blocks have the same band structure can alternatively be seen from the fact that the infinite periodic arrangements of either type of blocks are indistinguishable from the other.

Now we turn our attention to the source-sink behaviour. To more easily deal with the projective plane $\rp$, we introduce the following chart:
\begin{align*}
    \varphi: \rp &\to \overline{\R}\coloneqq\R \cup \{\infty\}\\
    (a,b) &\mapsto \frac{a}{b} .
\end{align*}

For a function $f:[a,b] \to \overline{\R}$ with $a<b$, we will use the notation $f:f(a)\nearrow f(b)$ on $[a,b]$ if $f$ \emph{increases monotonically} (with the possible exception of a jump from $+\infty$ to $-\infty$) from $f(a)$ to $f(b)$. Analogously, we write $f:f(a)\searrow f(b)$ on $[a,b]$ if $f$ \emph{decreases monotonically}.
\begin{figure}
    \centering
    \includegraphics[width=0.95\linewidth]{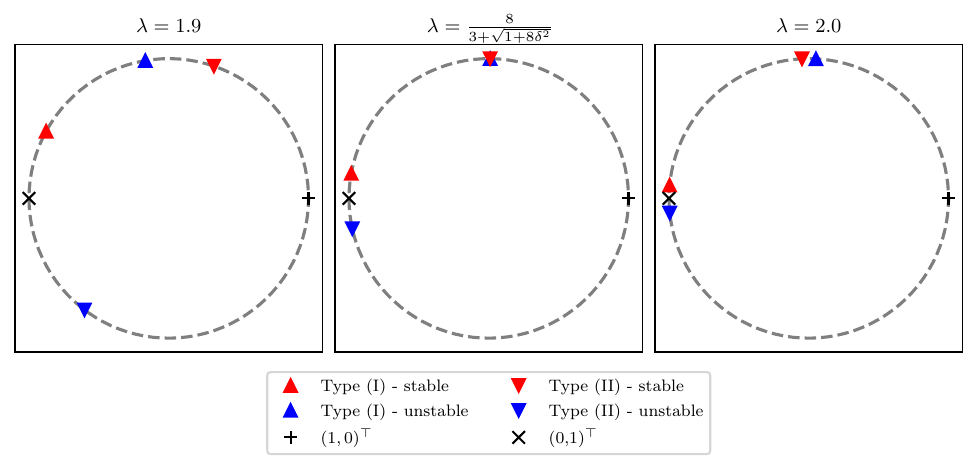}
    \caption{Source (stable) $s(\mc P^\lambda_d)$ and sink (unstable) $u(\mc P^\lambda_d)$ phase for both block types $d=I,II$ on the projective space $\rp\simeq S^1$ as $\lambda$ passes through the critical value $c_1$ (for $\delta=10^{-1}$). We observe that the SSC is violated for $\lambda<c_1$, that $u(\mc P_I^{\lambda})=s(\mc P_{II}^{\lambda})$ for $\lambda=c_1$ and that the SSC is fulfilled afterwards.}
    \label{fig:Source_Sink_Transition}
\end{figure}

For the type-$I$ block in the first bandgap $G_1=(g_1,g_2)\coloneqq[\frac{2}{1+\delta}, \frac{2}{1-\delta}]$, we now find that 
\[
 \lambda\mapsto\varphi(u(\mc P_I^\lambda)):\frac{1+\delta}{2}\nearrow+\infty 
 \quad\text{and}\quad
 \lambda\mapsto\varphi(s(\mc P_I^\lambda)):\frac{1+\delta}{2}\searrow-\infty.
\]
Conversely, for the type-$II$ block, we find that 
\[
 \lambda\mapsto\varphi(u(\mc P_{II}^\lambda)):-\infty \nearrow \frac{1-\delta}{2}
 \quad\text{and}\quad
 \lambda\mapsto\varphi(s(\mc P_{II}^\lambda)):+\infty \searrow \frac{1-\delta}{2}.
\]
Thus, the sources and sinks of $\mc P_{I}^\lambda$ and $\mc P_{II}^\lambda$ differ both in their orientation and start-/endpoints as $\lambda$ moves through $G_1$.

From their different orientations, it follows that the source-sink condition is violated in a neighbourhood of both endpoints of $G_1$. In particular, we can explicitly find that the source-sink condition is violated in the sub-intervals $[g_1, c_1]$ and $[c_2, g_2]$ of $G_1$, while it is satisfied in the interval $(c_1,c_2)$
where
\[
    c_1\coloneqq\frac{8}{3+\sqrt{1+8\delta^2}} \quad \text{and} \quad   c_2\coloneqq\frac{2}{1-\delta^2}.
\]
These critical points are exactly the crossover points 
\[
    u(\mc P_I^{c_1})=s(\mc P_{II}^{c_1})\quad \text{and} \quad s(\mc P_I^{c_2})=u(\mc P_{II}^{c_2}),
\]
where the source phase of one block aligns with the sink phase of the other.
This crossover at point $c_1$ is illustrated in Figure \ref{fig:Source_Sink_Transition}. We can see that the SSC is violated for $g_1\leq \lambda\leq c_1$ but fulfilled for $c_1\leq \lambda\leq c_2$. Notably, in contrast to the off-diagonal operator, this analysis also reveals that the interface modes at $c_1$ or $c_2$ are not necessarily separated from the rest of the spectrum, but merely either at the top or bottom of the band.

In Figure \ref{fig:SSC_violation}, the connection between these block-level properties and the spectrum $\text{Spec}(\mc C^\omega)$ of $\mc C^\omega$ is exhibited clearly. In particular, we can see that even though $G_1$ is a shared bandgap of both block types, the gap of the composite system is heavily polluted due to the  violation of the SSC. Nevertheless, in the interval $(c_1,c_2)$ in which the SSC is indeed fulfilled, the composite system \emph{does} exhibit a bandgap.
\begin{figure}
    \centering
    \includegraphics[width=0.5\linewidth]{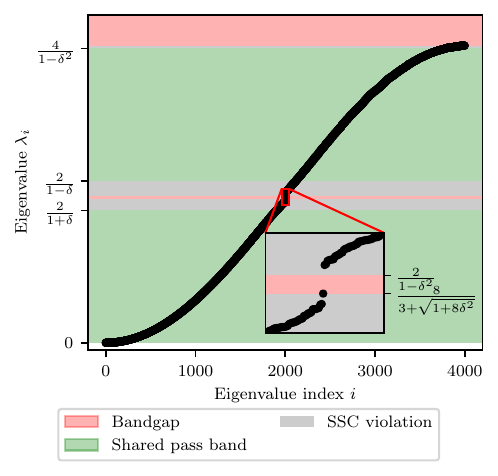}
    \caption{Spectrum $\text{Spec}(\mc C^\omega)$ for a left$-I$, right-$II$ dominated block disordered system ($K=2000, \delta=10^{-1})$, together with the distinct spectral regions determined by the block properties. Here, \emph{shared pass band} denotes $\lambda$ in the band of both blocks, \emph{SSC violation} denotes $\lambda$ in the shared gap but with violated SSC, and \emph{bandgap} denotes the bandgap of the composite system corresponding to $\lambda$ in the shared gap together with a fulfilled SSC.}
    \label{fig:SSC_violation}
\end{figure}

Finally, we are in a position to discuss the existence of the interface modes under differing domination regimes. 
The propagation matrix characterization yields a very straightforward explanation, as seen in the following lemma. 
\begin{lemma}\label{lem:propmatdecay}
Consider a left$-I$, right-$II$ dominated sequence $\omega$ and let $\lambda=c_1$ be such that $u(P_I^\lambda) = s(P_{II}^\lambda)$. Let $v\in u(P_I^\lambda) = s(P_{II}^\lambda)$ with $\norm{v}$=1. Then, there exist an $\varepsilon>0$, $\xi <0$, and $N\in \N$ such that, for any $n\in \Z$ with $\abs{n}>N$, we have
\[
    \norm{\mc P^\lambda_n(0) v} < \xi^{\varepsilon(n+1)}.
\]

The same result holds for a right$-I$, left-$II$ dominated sequence $\omega$ with $\lambda=c_2$ and $v\in s(P_I^\lambda) = u(P_{II}^\lambda)$.
\end{lemma}
\begin{proof}
    We shall consider $n>0$ without loss of generality. Due to their band-symmetry of the blocks we have that $\xi_s = \xi_u^{-1}<1$ where $\xi_s$ is the eigenvalue of $P_{II}^\lambda$ corresponding to the stable direction and $\xi_u$ is the eigenvalue of $P_{I}^\lambda$ corresponding to the unstable direction, and will thus simply denote by $\xi\coloneqq \xi_s$. 
    Due to the fact that $v\in u(\mc P_I^\lambda) = s(\mc P_{II}^\lambda)$, we have 
    \[
        \mc P_{II}^\lambda v = \xi v \quad \text{and} \quad \mc P_{I}^\lambda, v = \xi^{-1} v
    \]
    and therefore also 
    \[
        \norm{\mc P^\lambda_n(0)v} = \xi^{-C_I(n;\omega)}\xi^{C_{II}(n;\omega)} = \xi^{2C_{II}(n;\omega) - (n+1)}.
    \]
    But, from the right-$II$ domination of $\omega$, there must exist an $\varepsilon>0$ and $N\in N$ such that $C_{II}(n;\omega)>\frac{1+\varepsilon}{2}(n+1)$ for any $n>N$.
    Plugging this into the above equality yields \[
        \norm{\mc P^\lambda_n(0)v} < \xi^{\varepsilon(n+1)},
    \]
    as desired.
\end{proof}

In particular, the above lemma ensures that the sequence $j\mapsto \mc P^\lambda_n(0) v$ lies in $\ell^2(\Z)\otimes\R^2$ and can thus be transformed into a square-summable eigenmode of $\mc C^\omega$ via the cocycle cohomology described in \cite[Definition 3.8]{ammari.thalhammer.ea2025Uniform}.

\begin{figure}
    \centering
    \includegraphics[width=0.95\linewidth]{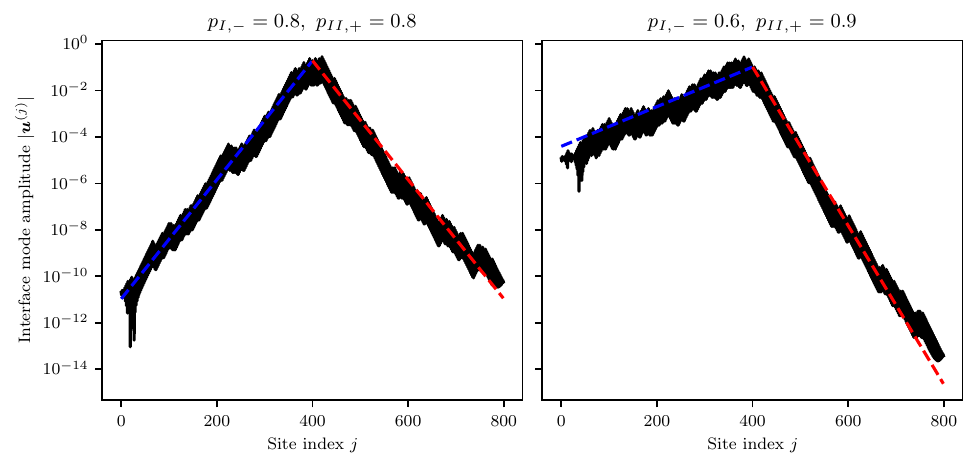}
    \caption{Interface mode for two block disordered left-$I$, right-$II$ dominated systems with varying domination strengths ($K=400, \delta=10^{-1}$). The expected left and right decay, as predicted by $\xi^{\varepsilon_-}$ and $\xi^{\varepsilon_+}$, is plotted as a blue dashed and red dashed line respectively.
    \textbf{Left:} Equal strength left and right domination. \textbf{Right:} Weak left and strong right domination.}
    \label{fig:Interface_Mode_Decay}
\end{figure}
\begin{remark}\label{rmk:Interface_Mode_Decay}
    Moreover, the above lemma also directly yields a precise decay estimate of the interface mode. Consider again a left$-I$, right-$II$ dominated sequence $\omega$ and let $\lambda=c_1$ and $\xi = \xi_s(\mc P_{II}^\lambda) = (\xi_u(\mc P_{I}^\lambda))^{-1}$ be the propagation matrix eigenvalue. From \cref{def_dominant_system}, we know that there must exist $\varepsilon_-, \varepsilon_+>0$ such that 
    \[
        \lim_{K\to\infty}\frac{C_I(-K;\omega)}{K+1}=\frac{1+\varepsilon_-}{2} \quad \text{and} \quad\lim_{K\to\infty}\frac{C_{II}(K;\omega)}{K+1}=\frac{1+\varepsilon_+}{2}.
    \]

    From the above lemma it then follows that the asymptotic decay rate of the interface mode in, respectively, the negative and positive direction must be $\xi^{\varepsilon_-}$ and $\xi^{\varepsilon_+}$. This prediction is illustrated in \cref{fig:Interface_Mode_Decay} for varying domination strengths. As predicted above, stronger domination leads to a more rapid decay of the interface mode. We control domination strength by sampling the block sequence $\omega$ independently and identically distributed from the set $\{I, II\}$ with probability $p_{I,-}, p_{II,-}$ in the left half and $p_{I,+}, p_{II,+}$ in the right half.
\end{remark}

\section{Quasi-periodic sequences}
\begin{figure}
    \centering
    \includegraphics[width=0.95\linewidth]{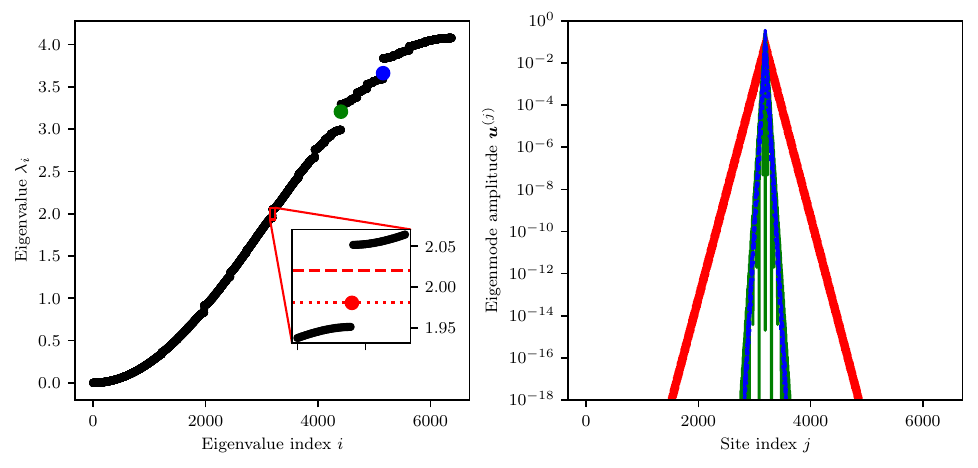}
    \caption{Spectrum and interface modes for a block disordered left-$I$, right-$II$ dominated system with a quasi-periodic sequence $\omega$ after $16$ Fibonacci iterations ($\delta=10^{-1}$). {Left:} Spectrum with \emph{multiple} interface modes. The classical interface eigenvalue, predicted by \cref{thm_interface_mode_finite_chain}, is marked in red and again lies at the lower critical point. However, we can see that the system now exhibits a range of additional bandgaps together with corresponding new interface eigenvalues (marked in blue and green). {Right:} Interface modes corresponding to the interface eigenvalues, coloured in line with their respective frequencies.}
    \label{fig:Fibonacci_Interface_Modes}
\end{figure}
Notably, while we considered sequences $\omega$ sampled \emph{i.i.d.} in the previous section, \cref{thm_interface_mode_finite_chain} holds for any dominated sequence -- including quasi-periodic ones. The canonical quasi-periodic sequence we consider is the \emph{Fibonacci tiling} (see, for instance, \cite{jagannathan2021fibonacci}) that can be constructed using a simple replacement rule as follows.
\begin{definition}[Fibonacci tiling]
We define the $k$\textsuperscript{th} \emph{$II$-dominated Fibonacci tiling} $\omega_k\in \{I,II\}^M$ recursively as $\omega_0 = (I)$ and $\omega_{k+1}$ obtained from $\mc \omega_{k}$ using the replacement rules $I \mapsto II$ and $II \mapsto II, I$. The \emph{$I$-dominated Fibonacci tiling} is defined from the $II$-dominated Fibonacci tiling under the exchange $I \leftrightarrow II$.    
\end{definition}
A signature characteristic of the $II$-dominated Fibonacci tiling is that the ratio of type-$I$ symbols to type-$II$ converges to $1:\phi$, where $\phi = \frac{1+\sqrt{5}}{2}$ is the \emph{golden ratio}. To be precise, we have 
\[
    \lim_{k\to\infty}\frac{C_I(M_k;\omega_k)}{M_k}=\frac{1}{\phi^2}<\frac{1}{2} \quad \text{and} \quad \lim_{k\to\infty}\frac{C_{II}(M_k;\omega_k)}{M_k}=\frac{1}{\phi}>\frac{1}{2},
\]
where $M_k$ denotes the \emph{length} of the $k$\textsuperscript{th} Fibonacci tiling. Therefore, it is indeed $II$-dominated in the sense of \cref{def_dominant_system}.

The Fibonacci tiling does not satisfy the \emph{pseudo-ergodicity} condition,  preventing a direct application of \cref{thm:saxonhutner}. 
However, from \cite[Proposition 3.3]{ammari.thalhammer.ea2025Uniform}, we know that 
\[
    \text{Spec}(\mc C^{\omega'}) \subset \text{Spec}(\mc C^{\omega})
\]
for any sequence $\omega'\in \{I,II\}^{\Z}$ and any pseudo-ergodic sequence $\omega\in \{I,II\}^{\Z}$. Consequently, any block disordered system arranged according to a Fibonacci tiling still retains the bandgap $(\frac{8}{3+\sqrt{1+8\delta^2}}, \frac{2}{1-\delta^2})$ guaranteed by \cref{thm:saxonhutner}, but may exhibit additional ones. 

In \cref{fig:Fibonacci_Interface_Modes}, we consider a left-$I$, right-$II$ dominated system constructed by concatenating a $I$-dominated Fibonacci tiling with a $II$-dominated Fibonacci tiling. In fact, this system exhibits a variety of additional bandgaps. Furthermore, \cref{thm_interface_mode_finite_chain} appears to hold also for these additional bandgaps, giving the corresponding interface modes.

\section{Concluding remarks}

In this paper, we have proved the existence of topologically protected interface modes in block disordered chains of resonator dimers of dominated type. Our results can be rigorously extended to other aperiodic systems such as bound length, hyperuniform, and quasiperiodic samplings, as done in \cite{ammari.barandun.ea2025Universal}. This would be the subject of a forthcoming paper. We also plan to extend our study to higher dimensions. In an on-going project \cite{qiu2025bec_disorder}, we develop a mathematical theory for a two-dimensional topological Anderson insulator, in which the disorder is present by random distribution of the gap-opening potential (\emph{e.g.}, the $\sigma_z$ term in the Qi-Wu-Zhang model \cite{qi2006qwz_model,ezawa2024nonlinear_phase_transition}). With the technique in \cite{qiu2025bulk,elgart2005shortrange+functional}, we can provide a topological characterization of the random-gap model and prove the bulk-edge correspondence principle on a finite sample.

\section*{Acknowledgments} This work was supported in part by the Swiss National Science Foundation grant number 200021-236472. It was completed while H.A. was visiting the Hong Kong Institute for Advanced Study as a Senior Fellow.

\footnotesize
\bibliographystyle{plain}
\bibliography{ref, zotero}

\end{document}